%% file: main.tex
\documentclass[letterpaper, 10 pt, conference]{ieeeconf}

\IEEEoverridecommandlockouts
\usepackage{cite}
\usepackage{graphics} 
\usepackage{epsfig} 
\usepackage{mathptmx} 
\usepackage{times} 
\usepackage{amssymb}  
\usepackage[cmex10]{amsmath}
\usepackage{amsthm}
\usepackage{amsfonts}
\usepackage{latexsym}
\usepackage{mathrsfs}
\usepackage{xcolor}
\usepackage{svg}
\svgpath{{../figures/}}
\usepackage[normalem]{ulem}
\usepackage{mathtools}
\usepackage{xspace}
\usepackage{geometry}
\usepackage{graphicx}
\usepackage{subcaption}
\newcommand{\rulesep}{\unskip\ \vrule\ }
\newcommand{\rev}[1]{\textcolor{black}{#1}}
\newcommand{\revf}[1]{\textcolor{black}{#1}}

\usepackage[linkbordercolor={1 1 1},citebordercolor={1 1
  1},urlbordercolor={0.0 0.0
  0.0},urlcolor=blue,colorlinks=true,linkcolor=black,citecolor=black]{hyperref}
\usepackage{url}

\theoremstyle{remark}

\theoremstyle{definition}
\newtheorem{defn}{Definition}

\theoremstyle{plain}

\newtheorem{thm}{Theorem}
\newtheorem{prop}{Proposition}  
\newtheorem{example}{Example}

\newcommand{\R}{\mathbb{R}}
\renewcommand{\S}{\mathit{S}}

\title{\LARGE \bf
Onboard Safety Guarantees for Racing Drones: \\ 
High-speed Geofencing with Control Barrier Functions
}

\author{Andrew Singletary, Aiden Swann, Yuxiao Chen, and Aaron D. Ames
\thanks{
Andrew Singletary, Aiden Swann, Yuxiao Chen, and Aaron D. Ames are with Department of Mechanical and Civil Engineering,
        California Institute of Technology, Pasadena CA 91125, U.S.A. Email addresses:
		{\tt \small \{asinglet, aswann, ychen, ames\}@caltech.edu}. This work is supported by AeroVironment and NSF CPS award \#1932091.
}
}

\begin{document}

\maketitle
\thispagestyle{empty}
\pagestyle{empty}

\begin{abstract}
This paper details the theory and implementation behind practically ensuring safety of remotely piloted racing drones. We demonstrate robust and practical safety guarantees on a 7" racing drone at speeds exceeding 100 km/h, utilizing only online computations on a 10 gram micro-controller. To achieve this goal, we utilize the framework of control barrier functions (CBFs) which give guaranteed safety encoded as forward set invariance.  To make this methodology practically applicable, we present an implicitly defined CBF which leverages backup controllers to enable gradient-free evaluations that ensure safety.  The method applied to hardware results in smooth, minimally conservative alterations of the pilots' desired inputs, enabling them to push the limits of their drone without fear of crashing. Moreover, the method works in conjunction with the preexisting flight controller, resulting in unaltered flight when there are no nearby safety risks. Additional benefits include safety and stability of the drone when losing line-of-sight or in the event of radio failure.
\end{abstract}

\section{INTRODUCTION}
\label{sec:introduction}
\input{sections/1_introduction}

\section{Preliminaries}
\label{sec:motivation}
\input{sections/2_motivation}

\section{Theory}
\label{sec:theory}
\input{sections/3_theory}

\section{Modeling and Implementation}
\label{sec:implementation}
\input{sections/4_implementation}

\section{Results}
\label{sec:results}
\input{sections/5_results}

\section{CONCLUSION}
\label{sec:conclusion}
\input{sections/6_conclusion.tex}

\renewcommand{\baselinestretch}{0.98}
\bibliographystyle{IEEEtran}
\bibliography{refs.bib}

\end{document}

%% file: sections/1_introduction.tex
As hobby drones become more and more capable, the interest in drone racing continues to increase. Transparency Market Research predicts the drone racing industry to reach a valuation of \$786m by 2027. And while progress is being made in autonomous racing drones \cite{foehn2020alphapilot,jung2018perception}, it is still an area dominated by humans \cite{delmerico2019we}.  The goal of this paper is to study racing drones in the theoretic context of achieving safe flight through minimal pilot interventions.  Importantly, we demonstrate this practically in a realistic scenario: high speed drone flight. 

Safety of small aerial vehicles is a heavily researched area.  These works generally focus on safely planning trajectories rather than intervening along a desired trajectory.  In this setting, \cite{tordesillas2019real} accounts for the low computational ability of drones, as well as the slow updates of mapping software, in their design of a planner for quick flight in unknown environments using motion primitives. While this and similar planners \cite{tordesillas2019faster} have demonstrated results in unknown environments, they have not been demonstrated at the high speeds seen in drone racing. This is true for nearly all vision planners \cite{kaufmann2018deep,kaufmann2019beauty}, as the localization and mapping algorithms simply cannot keep up with speeds that human operators are capable of. \rev{Moreover, for known environments, reinforcement learning has been utilized to plan highly dynamic trajectories at speeds exceeding 60 km/h \cite{song2021autonomous}, but attempting to track these trajectories on hardware results in large tracking errors.} This points to the difficulty of adapting existing strategies to ensure safety with human operators in the loop. 

While most drone racing research focuses on autonomy \cite{moon2019challenges}, this work departs from this paradigm with the goal of giving the human operators as much freedom and control authority as possible subject to safety constraints.
In particular, we seek to guarantee safety of the drone in known environments (with realistic sensing and actuation constraints) through minimal operator control intervention.  This can be viewed as a ``safety filter'' that allows the pilot to aggressively operate the drone in a safe fashion---even at high speeds and when performing aggressive maneuvers. 

\begin{figure}
    \centering
    \includegraphics[width=1\columnwidth]{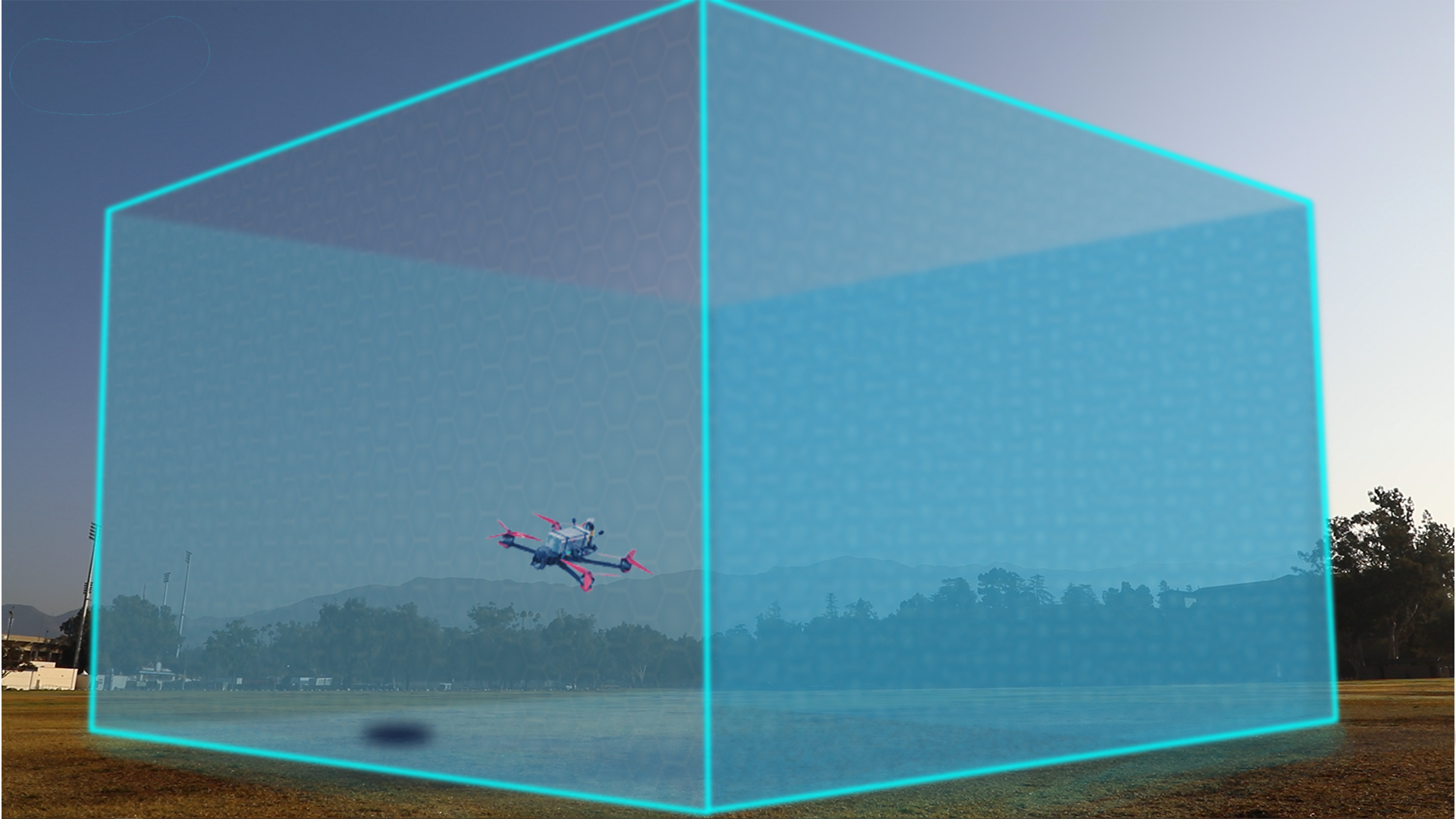}
    \caption{An illustration of geofencing---the pilot can move freely in the boxed region but cannot leave this volume---which is enforced during high speed flight (upwards of 100 km/h) through the methodology presented in this paper.
 }
    \label{fig:intro}
\end{figure}

\rev{While the concept of minimally invasive, shared control systems is less studied than its fully autonomous counterpart, several approaches exist in the literature. In \cite{broad2019highly}, the authors propose a sampling-based MPC approach that generates many possible safe trajectories at each time-step, and chooses the one closest to the user's desired input subject to safety conditions. Another MPC-based approach is demonstrated in \cite{tearle2021predictive}, which uses learning to minimize conservatism, but neither of these approaches are able to run in real-time on a microcontroller. In the context of geofencing, \cite{zhang2017model} presents an MPC-based approach, but it lacks the guaranteed feasibility of \cite{broad2019highly}. In \cite{hermand2018constrained}, an Explicit Reference Governor (ERG) scheme modifies the derivative of the applied inputs subject to safety constraints utilizing Lyapunov functions. While this approach is optimization-free and could be implemented online, it is difficult to find the required upper-bound of the Lyapunov function that guarantees constraint satisfaction.}

\rev{Control barrier functions \cite{ames2017cbf}, have recently been proven} to provide an effective means of enforcing safety on a wide variety of robotic systems \cite{ames2019control}, (including drones \cite{wang2018safe}).
However, state-of-the-art barrier function implementations are not particularly well-suited for high-speed drone flight as they typically rely on simplified models that become invalid at high speeds, or have no guarantees of feasibility.  The backup-set approach can be used to enforce safety of quadrotors on the full-order dynamics \cite{chen2020guaranteed}, but requires more onboard computational power than is available on small racing drones.

The goal of this paper is to provide formal guarantees of safe flight via minimal operator intervention for high speed drones, utilizing only onboard sensing and computation. 
To achieve this goal, we leverage the \emph{implicit control barrier functions} framed in the context of backup controllers.  
The concept of using implicit control barrier functions to guarantee safety in a more computationally efficient manner, without the use of derivatives, was first introduced in \cite{singletary2020safety}. \rev{While the method presented in this paper is based on our previous work \cite{singletary2020safety}, we have made significant improvements on the algorithm to enhance its performance and make it more applicable to hardware implementation. The previous algorithm had a tendency to get stuck near the boundary of the safe set if model-mismatch or large disturbances were present, and thus it was tested only in simulation. The first contribution of this paper is a more refined safety filter that avoids issues near the set boundaries and interfaces with any off-the-shelf flight controller, dramatically extending the \rev{practicality} of the method while still retaining the formal safety guarantees.
The second contribution is our extensive, real-world testing of this safety filter. To this end, \emph{we experimentally realize control barrier functions to enforce geofencing (cf. Fig. \ref{fig:intro}) at speeds of 100 km/h.}}


%% file: sections/2_motivation.tex
\subsection{Safe flight and set invariance}

To prevent crashing and guarantee safety, our goal is to ensure that the system's state $x(t)$ stays in a predefined safe set $\S$, such as the box shown in Figure \ref{fig:intro} typically seen in geofencing. Before formalizing this notion of safety, we must first introduce some notation and definitions.

We consider a nonlinear control-affine dynamic system:
\begin{align} \label{eqn:dyn}
    \dot{x} = f(x) + g(x) u,
\end{align}
where $x \in \mathbb{R}^{n}$ is the state, $u \in U$ is the control input, to be chosen from an admissible input set $\mathit{U} \subseteq 
\mathbb{R}^{m}$, while $f: \mathbb{R}^{n} \to \mathbb{R}^{n}$ and $g: \mathbb{R}^{n} \to \mathbb{R}^{n \times m}$ describe the dynamics.
Given a controller $u = \rho(x)$ and an initial condition ${x(t_0) = x_{0} \in \mathbb{R}^{n}}$, the solution to this system is given by the flow map ${x(t)=\Phi_{\rho}(x_{0},t)}$, $t \geq t_0$, i.e., $\Phi_{\rho}$ takes the initial state and time of flow and maps to the future state of the closed-loop system under $\rho$. 

The goal of forcing the system to remain in $\S$ at all time can be achieved through the concept of set invariance. 

\begin{defn}
A set $\mathit{S}$ is called \textit{invariant} if the system state stays in the set indefinitely, i.e. $\forall t \geq t_0,\ x(t) \in \mathit{S}$. $\mathit{S}$ is \textit{control invariant} if there exists a control law $k:\mathbb{R}^n\to\mathit{U}$ mapping $x$ to the admissible input set $\textit{U}$ such that the system is invariant under $k$, i.e. $\forall t \geq 0, \forall x_0\in\mathit{S},\ \Phi_{k}(x_{0},t)\in\mathit{S}$.
\end{defn}

\subsection{\rev{Control barrier functions}}

\rev{Given an invariant set $S$, a control barrier function can be formulated and used to enforce set invariance at each time-step.}

\begin{defn}[\hspace{-.1pt}\cite{ames2017cbf}]
\label{def:cbf}
Let $\S \subset \R^n$ be the set defined by a continuously differentiable function $h: \R^n \to \R$:
\begin{eqnarray}
\S & = & \{ x \in \R^n ~ : ~ h(x) \geq 0 \} , \nonumber\\
\partial \S & = & \{ x \in \R^n ~ : ~ h(x) = 0 \}, \nonumber\\
\mathrm{Int}(\S) & = & \{ x \in \R^n ~ : ~ h(x) > 0 \}. \nonumber
\end{eqnarray}
Then $h$ is a \textit{control barrier function (CBF)} if $\frac{\partial h}{\partial x}\neq 0$ for all $x\in\partial \S$ and there exists an \emph{extended class $\mathit{K}$ function} (\cite[Definition 2]{ames2017cbf})
$\alpha$ such that for all 
$ x \in \S$, $\exists u$ s.t. 
\begin{align}
\label{eqn:cbf:definition}
\underbrace{\frac{\partial h}{\partial x} f(x) + \frac{\partial h}{\partial x} g(x) u}_{\dot{h}(x,u)} \geq - \alpha(h(x)).
\end{align}
\end{defn}

\rev{By choosing an input $u$ such that \eqref{eqn:cbf:definition} holds, the invariance of $S$ is guaranteed \cite{ames2017cbf}. Therefore, given an}
input from the flight operator, or another preexisting controller, $u_{\textrm{des}}(x,t)$, we can pick the closest input $u^*(x,t)$ to $u_{\textrm{des}}(x,t)$ that guarantees safety by solving the following quadratic program (QP):
\begin{align}
\label{eqn:QPsimple}
u^*(x,t) = \underset{u \in \mathit{U}}{\operatorname{argmin}} & ~  ~ \| u - u_{\rm des}(x,t) \|^2 \\
\mathrm{s.t.} & ~  ~ \frac{\partial h}{\partial x} f(x) + \frac{\partial h}{\partial x} g(x) u \geq - \alpha (h(x)). \nonumber
\end{align}
\rev{As mentioned in the introduction, CBFs work very well for maintaining set invariance. However, they are often restricted in their usage due to the difficulty in constructing an invariant set $S$.} Several methods for computing such sets exist \cite{aubin2009viability,mitchell2005time}, but almost all methods for nonlinear systems suffer heavily from the curse of dimensionality and are inapplicable to drones without model simplification.




\subsection{\rev{Backup controllers to construct invariant sets}}

The backup set approach \cite{gurriet2018online} was recently proposed as a remedy to the difficulty of obtaining control invariant sets. This approach \revf{assumes} the knowledge of a very small backup set $\S_B$ that is invariant under some backup controller $\pi(x)$. 
This can be, for example, a set describing the hovering maneuver for the drone and a controller commanding the drone to stop and hover in place; we will give more details on the backup set and controller selection later.

With the small backup set defined, a point $x_0 \in S$ is in the control invariant set $S_I$ if the system is able to reach the backup set $\S_B$ over some time $T$ without leaving the safe set $S$ under the predefined backup controller \cite{gurriet2018online}. This makes the description of our control invariant set:
\begin{equation} \label{eqn:backup_SI}
    S_I = \left\{ 
    x_0 \in \R^n \left| \underset{\forall \tau \in [0,T]}{\Phi_{\pi}(x_0,\tau)} \in \S \wedge \Phi_{\pi}(x_0,T) \in S_B\right.
    \right \}
\end{equation}
where $\Phi_{\pi}(x_0,t)$ is the flow of the system from initial condition $x_0$ after time $t$, governed by $\dot{x} = f(x) + g(x) \pi(x) = f_{\pi}(x)$, where $\pi(x)$ is the backup control policy. 

The concept of obtaining an invariant set $\S_I$ from $\S$ using backup controllers is illustrated in Figure \ref{fig:backup_set_illustrated}. 

\begin{figure}
    \centering
    \includegraphics[trim={0 0 20cm 0},clip,width=.9\columnwidth]{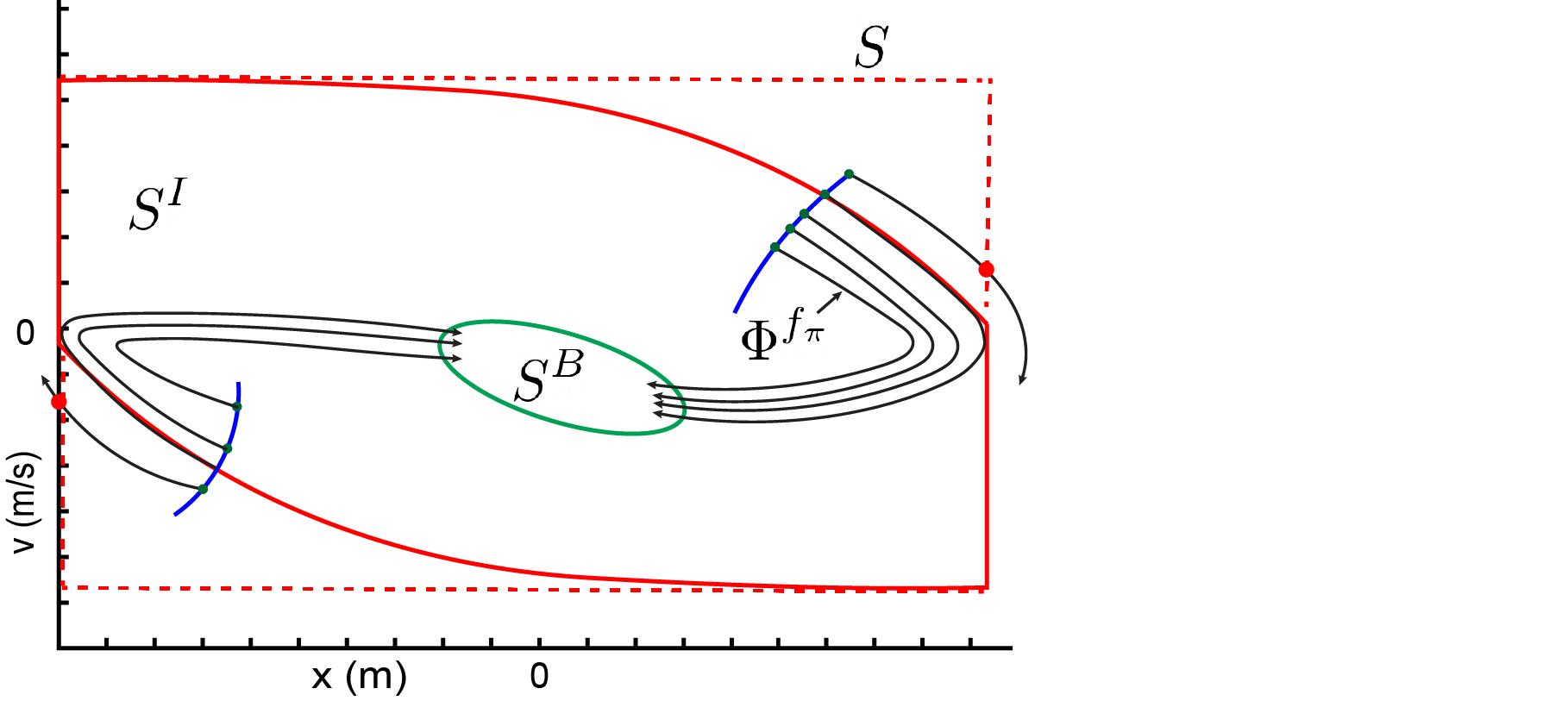}
    \caption{The flow $\Phi_{\pi}$ and the backup set $S_B$ are used to construct the invariant set $S_I$, following \eqref{eqn:backup_SI}\revf{.}}
    \label{fig:backup_set_illustrated}
\end{figure}

As shown in \cite{chen2021backup}, the set generated by a proper backup controller can approach the size of the maximal control invariant set. However, for drone racing, where the extraordinary high-speed is coupled with limited onboard computational ability, \rev{these types of CBFs} are not a realistic option. \rev{This is due to the computationally expensive gradient computation at each time-step to determine $\frac{\partial h}{\partial x}$}.

\rev{The following section introduces our proposed method, which is similar to backup-set CBFs, but does not require any gradient computations or optimization solvers.}

%% file: sections/3_theory.tex
\subsection{Smooth safety filtering along the backup controller}

In the previous section, we established that a backup controller and the corresponding backup set can be used to define a control invariant set as in \eqref{eqn:backup_SI}. Let us describe the safe set $\S$ and the backup set $S_B$ as the zero super-level sets of smooth functions $h(x)$ and $h_B(x)$. We rewrite $S_I$ defined in \eqref{eqn:backup_SI} as:
\begin{equation} \label{eqn:backup_SI_h}
    S_I = 
    \left\{ 
    x \in \R^n \left| \underset{\forall \tau \in [0,T]}{\left( h\left( \Phi_{\pi}(x,\tau)\right) \geq 0 \right)} 
     \land \left( h_B\left(\Phi_{\pi}(x,T)\right) \geq 0 \right)\right.
    \right \} .
\end{equation}
Furthermore, following \cite{chen2021backup}, we know that 
\begin{equation} \label{eqn:backup_SI_min}
    S_I = \left\{ 
    x \in \R^n \left|  \underset{\tau \in [0,T]}{\min} \left \{ h\left( \Phi_{\pi}(x,\tau)\right)
     , h_B\left(\Phi_{\pi}(x,T)\right) 
     \right \} 
     \geq 0\right.
    \right \} .
\end{equation}
For ease of notation, we define the following implicit CBF.
\begin{defn}
Given a control system \eqref{eqn:dyn} with flow map $ \Phi_{\pi}(x,T)$ associated with a backup controller $u = \pi(x)$, 
the \emph{implicit control barrier function (ImCBF)} $h_I$ is defined as
\begin{equation}
h_I(x) := \underset{\tau \in [0,T]}{\min} \left \{ h\left( \Phi_{\pi}(x,\tau)\right)
     , h_B\left(\Phi_{\pi}(x,T)\right) 
     \right \}.
\end{equation}
Associated with the ImCBF is the safe set given as in \eqref{eqn:backup_SI_min}: $S_I = \left\{ 
    x \in \R^n \mid h_I(x) \geq 0 \right\}$. 
\end{defn}
The ImCBF $h_I$ is implicitly defined as it relies on the flow map of the system, which is also implicit. $h_I$ can be evaluated by forward simulating the system under the backup controller.


While invariance could be enforced using the QP \eqref{eqn:QPsimple} with the ImCBF $h_I$ as its constraint, called backup CBF QP, the implicitly defined CBF requires additional computation power to solve. Typically, the backup CBF QP is not achievable in real-time without a desktop-grade CPU. \rev{A simplification must be made in order to guarantee safety with a low-weight, low-power microcontroller suitable for a small racing drone, such as the Teensy 4.}

Instead, we choose an approach utilizing the backup controller directly. Rather than switching to the backup controller when the state is about to leave $\S_I$, our approach smoothly switches between the desired controller and the backup controller as the boundary of $\S_I$ is approached. We define the function that regulates this smooth switching the \textit{regulation function}, and denote it $k(x,h_I(x))$.

The proposed method for safety filtering using the regulation function is formalized in the following proposition.

\begin{prop}[\hspace{-.1pt}\cite{singletary2020safety}]
Consider the open-loop system \eqref{eqn:dyn} under a continuous control law given by $k \left(x,h_I(x)\right) $. If 
\begin{equation} \label{eqn:lambda}
   k (x,0) = \pi(x) \quad \forall x \in S_I,
\end{equation}
for backup control law $\pi(x)$, then the closed-loop system under $k \left(x,h_I(x)\right)$ is forward invariant, i.e. safe.
\end{prop}
\begin{proof} 
By definition, $S_I$ is invariant under the control law $\pi(x)$. Therefore, $\forall x_0 \in S_I, \forall t \geq t_0, \Phi_{\pi}(x_0,t) \in S_I$\revf{.} By continuity of the flow operator, we know that $\Phi_{k}$ is continuous, and $h_I$ is continuous since it is the composition of continuous functions $h$ and $h_B$ and the flow $\Phi_{k}$. Therefore, the state must pass through $h_I(x) = 0$ before exiting $S_I$. Without loss of generality, label any such point $x_{h_0}$. Since $f_{k} = f_{\pi}$ at such a point, the system will remain in $S_I$ for all time, as $x_{h_0} \in \S_I \implies  \Phi_{\pi}(x_0,t) \in S_I, \forall t \geq t_{h_0}$.
\end{proof}

Additionally, to improve performance, we require that for $h_I(x) \gg 0$, $k \left(x,h_I(x)\right) = u_{\textrm{des}}(x)$. This way, the filter does not modify pilot's actions unless there is an eminent risk of leaving the safe set.

Other than the reduced computational complexity, there are several advantages to this filtering approach. Unlike control barrier functions, there is no notion of relative degree for the input, giving complete freedom in the choice of $h(x)$. Moreover, input bounds can be handled trivially as, by \revf{design, we choose} $\pi(x)\in\mathit{U}$, and $k$ can be constructed so that the input bound is always satisfied. Another benefit of this method is showcased in the following proposition.

\begin{thm} \label{prop:cont}
If $k(\cdot,\cdot)$ is locally Lipschitz in its arguments, and the dynamics under the backup controller $f_\pi$ are continuous and bounded on $\S$, then the resulting filter $k\left(x,h_I(x)\right)$ is locally Lipchitz continuous.
\end{thm}
\begin{proof}
By the definition of a CBF, $h(x)$ and $h_B(x)$ are differentiable. Moreover, the flow of the system $\Phi_{\pi}(x,t)$ is differentiable for continuous dynamics and backup controller by the second fundamental theorem of calculus, and therefore locally Lipschitz (since it is also bounded). Since Lipschitz continuity is preserved under the $\min$ operator, we have that $h_I(x)$ is locally Lipschitz continuous. Lastly, since Lipschitz continuity is preserved under compositions, we have that $k\left(x,h_I(x)\right)$ is locally Lipschitz continuous.
\end{proof}

\revf{While switching controllers can result in discontinuous inputs, the locally Lipschitz property demonstrated in Theorem \ref{prop:cont} guarantees smoother control inputs with this method.}

\subsection{Choice of regulation function}

The two requirements for our filtering function is that (i) the backup controller ${\pi(x)}$ is applied when $h_I(x) = 0$, (ii) it is Lipschitz continuous in its arguments. When $h_I(x) > 0$, we want the filter to mimic $u_{\textrm{des}}(x)$ as much as possible. To achieve this, we first choose the mixing function
\begin{equation} \label{eqn:smoothalpha}
    k \left(x,h_I(x)\right) = \lambda\left(x,h_I(x)\right) u_{\textrm{des}}(x) + (1-\lambda(x,h_I(x)))\pi(x),
\end{equation}
for $\lambda(x,h_I(x)): \R^n \times \R \to [0,1]$. Figure \ref{fig:lambda} illustrates how this mixing works. \rev{When $\lambda(x,h_I(x)) = 1$, the human operator is in complete control.} As the drone approaches the wall, the value of $\lambda(x,h_I(x))$ decreases, and it begins to slow down. Finally, near the boundary, a steady-state is reached between the operator's control action and the backup control action, and the drone stops. \revf{It is important that the backup controller attempts to move away from the boundary, so that the system does not get "stuck" near the boundary, i.e. $\lambda > 0$ always.}

\begin{figure}[t]
    \centering
    \includegraphics[width=\columnwidth]{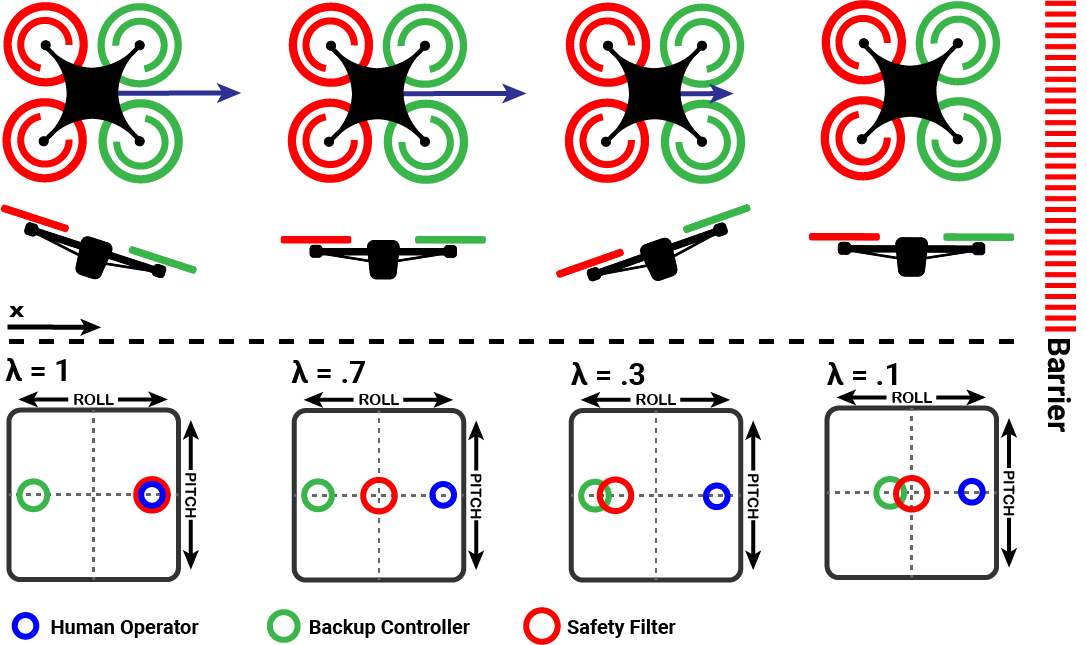}
    \caption{As the drone approaches the barrier, $\lambda$ decreases, resulting in the backup controller being utilized more.}
    \label{fig:lambda}
\end{figure}

To allow for maximum freedom of the operator, the function should exactly match $u_{\textrm{des}}(x)$ when $h_I(x) \gg 0$. One way to achieve this is by exploiting the exponential function:
\begin{equation}\label{eqn:smoothlambda}
    \lambda(x,h_I(x)) = 1-\exp\left(-\beta{h^+_I(x)}\right),
\end{equation}
where constant $\beta$ is used to tune how quickly the function $\lambda(x,h_I(x))$ decays, and $h^+_I(x) = \max (h_I(x),0)$ ensures that $\lambda(x,h_I(x)) \in [0,1]$. With this filtering controller, we get the desired behavior of $\lambda(x,h_I) \approx 1$ when $h_I(x) \gg 0$, while providing a smooth decay to 0 when as $h_I(x) \to 0$. The constant $\beta$ is used to tune how quickly the function $\lambda$ decays.
 
\subsection{Comparison to backup CBF condition}

To demonstrate the effectiveness of this method, we compare its performance to the backup set CBF approach \cite{gurriet2018online}. The two relevant metrics for this comparison are the computational times and the conservatism, as both methods provide guarantees of safety. 

While no perfect comparison can be made, due to the freedom in the selection of both $\lambda(x)$ and $k(x)$, these functions are chosen independently to result in smooth transitions when approaching the boundary of the set, while minimizing conservatism. 

\begin{figure}
    \centering
    \includegraphics[trim={1cm 0.25cm 1.5cm .65cm},clip,width=1\columnwidth]{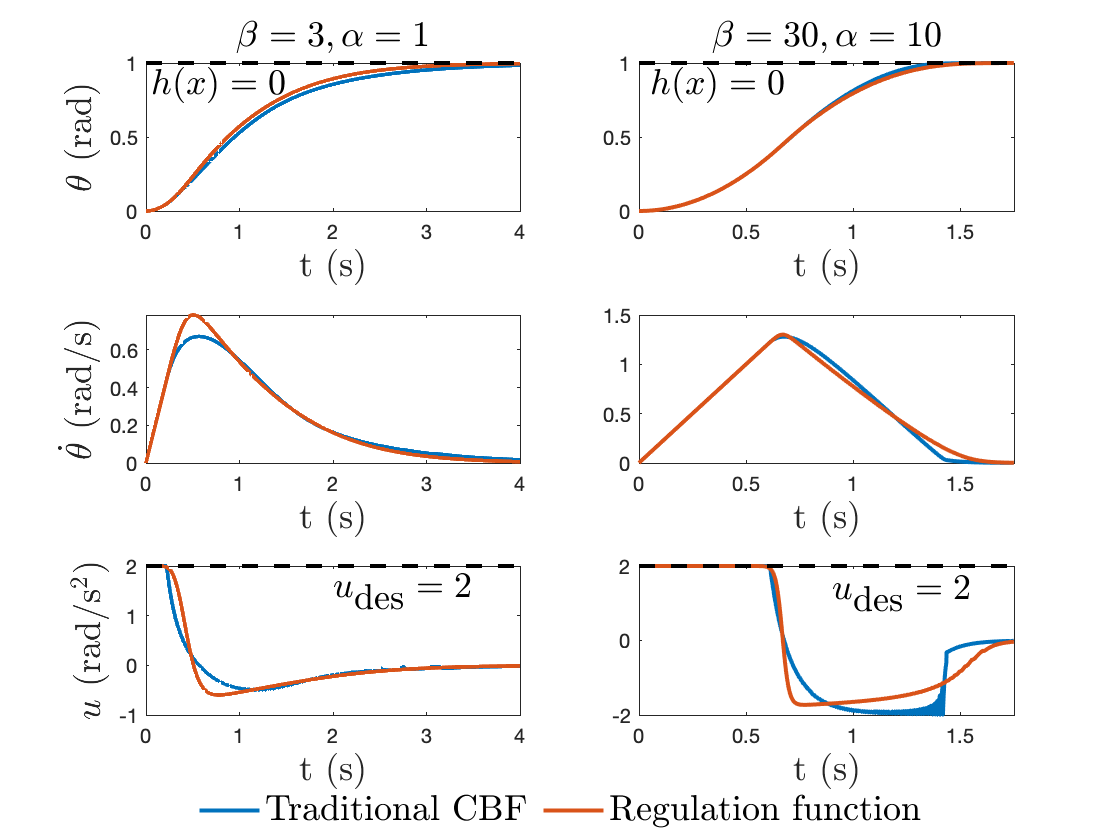}
    \caption{The filtering performance of the traditional CBF compared to the proposed regulation function, for two parameters $\beta$ in \eqref{eqn:smoothlambda} and the scalar $\alpha$ for the CBF \eqref{eqn:QPsimple}\revf{.}}
    \label{fig:cbf_comparison}
\end{figure}

\begin{example}
Consider an inverted pendulum with state $x$ dynamics $\dot{x}$
\begin{equation}
    x = 
    \begin{bmatrix}
    \theta \\
    \dot{\theta}
    \end{bmatrix} +
    \begin{bmatrix}
    0 \\
    1
    \end{bmatrix}
    u, \quad
    \dot{x} = 
    \begin{bmatrix}
    \dot{\theta} \\
    \sin(\theta)
    \end{bmatrix} +
    \begin{bmatrix}
    0 \\
    1
    \end{bmatrix}
    u,
\end{equation}
and backup control law 
\begin{equation}
\pi(x) = -Fx,
\end{equation}
which attempts to stabilize the system to the backup set 
\begin{equation}
h_B(x) = \min\left \{ \left(\frac{\pi}{12}\right)^2 - x_1^2, \delta^2-x_2^2 \right \},
\end{equation}
for a small velocity value $\delta$ chosen to be 0.1 rad/s, while staying in the set
\begin{equation}
h(x) = \min\left \{ 1 - x_1^2, 2-x_2^2 \right \}.
\end{equation}
\end{example}

For the exact formulation of the backup-set CBF, see \cite{chen2021backup}. The filtering function used here is the same as that used on the drone \eqref{eqn:smoothalpha}, to be detailed in the following section. The inverted pendulum was commanded a constant angular acceleration of 2 rad/s$^2$, and the resulting positions, velocities, and filtered inputs are displayed in Figure \ref{fig:cbf_comparison}.

Two benefits of the smooth filter can be seen in this comparison. While filtering performance is similar, the regulation function is an order of magnitude faster to evaluate. Moreover, when the gains are increased to allow a rapid approach of the boundary of the safe set, the CBF begins to oscillate near the boundary, whereas the smooth filter does not suffer from such behavior. \rev{These oscillations occur due to numerical instability of the optimization problem as the system pushes against boundary of the safe set.}

It is important to note that the optimization-based CBF still has a distinct advantage in some situations. Utilizing gradient information allows quick motion along the boundary of the set, whereas with this switching approach, the value of $\lambda(x,h_I(x))$ will be low, limiting performance near the boundary. In future work, we will explore how this ability can be utilized in a derivative-free approach. \rev{While it is not particularly important for geofencing, it would be important in the context of collision avoidance while drone-racing.}

\begin{figure*}[t]
    \centering
    \begin{subfigure}{.49\textwidth}
    \includegraphics[trim={1cm 0 1.5cm 0},clip,width=\columnwidth]{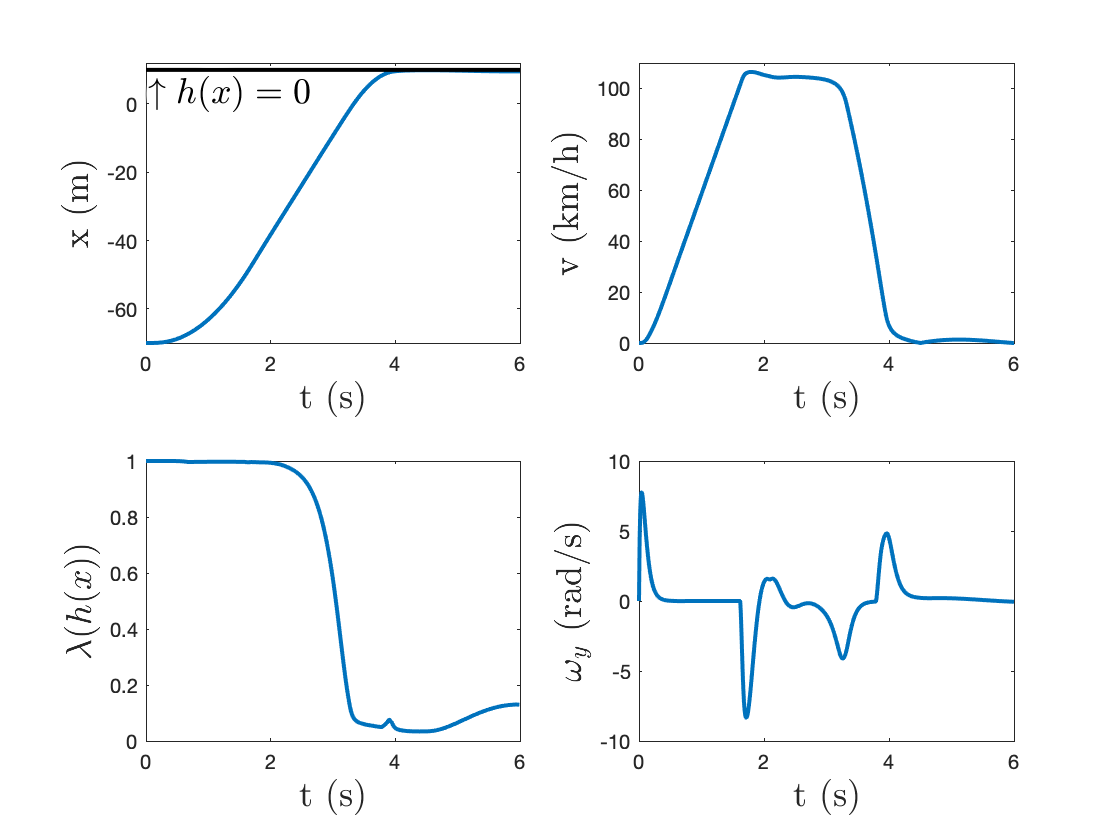}
    \label{fig:sim_x}
    \caption{Horizontal barrier active at 107 km/h. }
    \end{subfigure}
    \rulesep
    \begin{subfigure}{.49\textwidth}
    \includegraphics[trim={1cm 0 1.5cm 0},clip,width=\columnwidth]{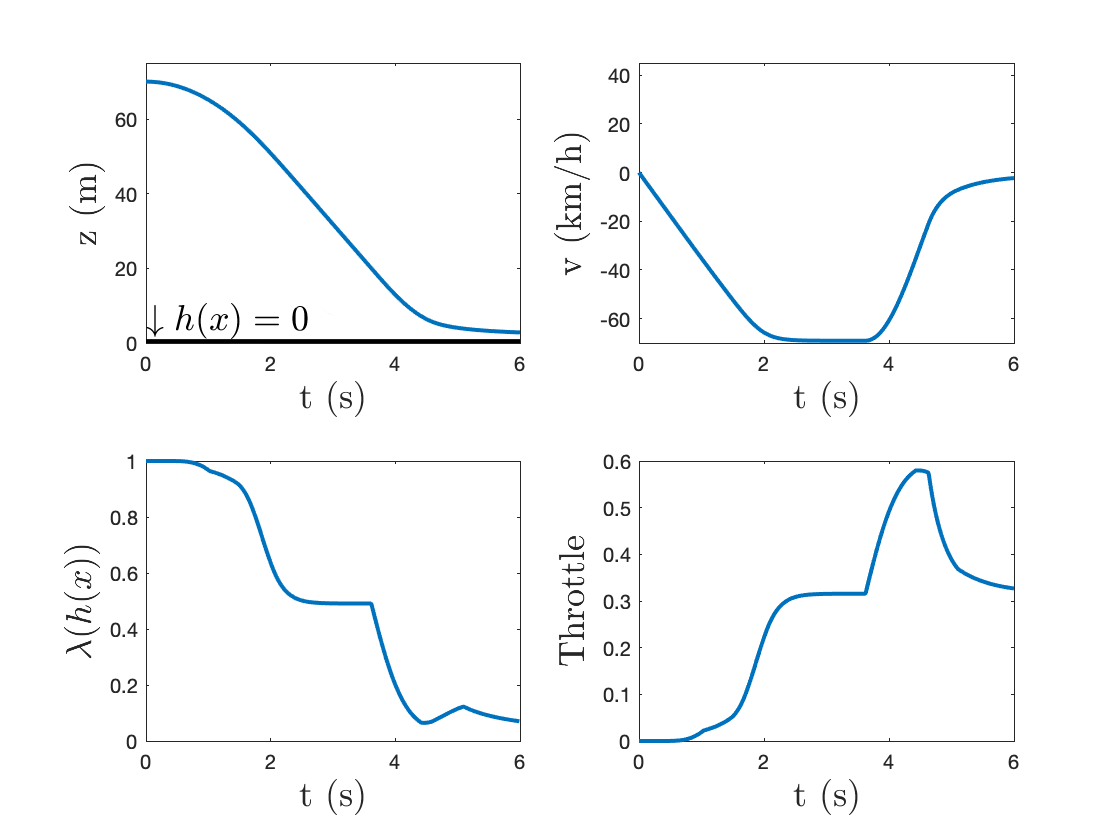}
    \label{fig:sim_z}
    \caption{Vertical barrier active at -70 km/h. }
    \end{subfigure}
    \caption{Simulation results of the two primary hardware test cases. On the left, the drone accelerates towards the barrier at $x_w = 10$ m. On the right, the drone free-falls from 70 m towards the barrier \rev{at $z_w = 0.5$ m.}}
    \label{fig:barriers_sim}
\end{figure*}

%% file: sections/4_implementation.tex
\subsection{Modeling the Drone and Onboard Flight Controller}

Flight controllers on modern racing drones are able to track desired angular rates extremely well. With the availability of low-cost, high-speed electronic speed controllers (ESCs) and rate gyros, state-of-the-art controller\revf{s} can track angular rate\revf{s} at control frequencies of 8 kHz.
\rev{We utilize this by wrapping our controller around the closed-loop system of the drone with the onboard flight controller.} Therefore, rather than the control inputs being the torques of the four motors, we command throttle and angular rates. This choice of architecture greatly simplifies the task of modeling the drone dynamics, and allows us to better filter the system in a way that minimizes the impact on the pilot. Moreover, the same filter can be applied to different drones with different dynamics, including those with six or eight rotors.

By modeling the response of the system to desired angular rate commands, the proposed method does not rely on perfect angular rate tracking. Through this, we also account for any delay in the filter stemmed from communication between the sensors and the onboard flight controller.

The drone and flight controller system is modelled as rigid body motion in the Special Euclidean Group in 3 dimensions, SE$(3)$. The state-space model $x \in \R^{13}$ is chosen to be
$  x = [
    p_w ,
    q ,
    v_w ,
    \omega_b]^T
    $
where $p_w = [x,y,z]^T$ is the position in the world frame, $v_w = [v_x,v_y,v_z]^T$ are the world-frame velocities, $q$ is the quaternion representation of the orientation with respect to the world frame, and $\omega_b = [\omega_x,\omega_y,\omega_z]^T$ are the body-frame angular velocities.

To model the system's response to an angular rate command, we set the derivative of the angular rates to
\begin{equation} \label{eqn:omega_model}
    \dot{\omega} = C(x)(\omega_{\textrm{des}}-\omega),
\end{equation}
where $C(x):\R^n \to \R_+$ is a (potentially state-dependent) function that determines how quickly the desired rates are tracked. For a well-tuned racing drone with minimal filter delay, its value should be on the order of 50, and can be treated as state-independent.

Lastly, to map the throttle command to thrust, we fit a second-order polynomial with data from the accelerometer and GPS. While this mapping will be dependent on the voltage, the inclusion of an integral term in the altitude controller is generally sufficient to eliminate any drift.

\subsection{Safe sets and backup controllers}

The primary goal of this work is to constrain the position of the drone inside a large polytope in the 3D space, inside of which the pilot has almost complete control, but is unable to leave.
To this end, we define the safe set
\begin{equation}
    h(x) = \min \left\{r_x^2 - (x-x_c)^2, r_y^2 - (y-y_c)^2, r_z^2 - (z-z_c)^2\right\},
\end{equation}
which is positive inside of a box with side lengths $(r_x,r_y,r_z)$ centered at $(x_c,y_c,z_c)$, and negative outside.

The backup controller $\pi(x)$ is a velocity controller on SE(3), inspired by \cite{5717652}. The backup controller attempts to bring the drone to zero velocity in the $x,y,z$, but has one other goal which is very important: to bring the drone away from the boundary if it is too close. This is critical, as if the drone were to simply stop at the boundary,
the pilot would be stuck at the edge of the safe set due to $\lambda(x)$ approaching 0.
To achieve this, we set the desired velocity to 
\begin{equation}
    v_x = \begin{cases}
    0 & r_x^2 - (x-x_c)^2 \geq \delta, \\
    -(\delta - r_x^2 + (x-x_c)^2) & \textrm{otherwise}. \\
    \end{cases}
\end{equation}
Under this backup controller, the drone will move a distance $\delta$ from the boundary before stopping. The desired velocities are identical for $y$ and $z$ directions.

Finally, the backup set $S_B$ is defined by the function
\begin{equation}
    h_B = -\sqrt{v_x^2+v_y^2+v_z^2}+\epsilon.
\end{equation}

This backup set ensures that the drone is able to slow itself to a speed of $\epsilon$, chosen to be $0.1$ m/s, \rev{thus guaranteeing that the drone is able to stop before hitting the boundary.}

\subsection{Modification of safety filter for very high speeds}

Modifications must be made to the function $\lambda(x,h_I(x))$ to work well at very high-speeds. This is because $\beta$ must be made large to have smooth breaking at high speeds, which would make the filter overly conservative near the boundary at low speeds. This can be fixed simply by scaling the value of $h_I(x)$ by the inverse of the velocity towards the barrier. The safety filtering function used by the drone is

\begin{equation}\label{eqn:smoothlambda_new}
    \lambda(x,h_I(x)) = 1-\exp\left(-\frac{\beta h_I(x)^+}{v_\perp^{+}}\right),
\end{equation}
where $v_\perp$ is the velocity in the direction of the barrier. 
 

\subsection{Simulation}

Before testing the barrier functions on hardware, we first devise the test cases in simulation. While the safety filters are the same in simulation and on hardware, the hardware is operated by a human pilot, while the sim has its own desired controllers, so some discrepancies will arise because of this.

Two primary test cases were run in simulation, a high-speed horizontal test, with the goal of successful filtering at 100 km/h, and a free fall from 70 m. The results of the horizontal simulation are shown in Fig. 5a. The drone accelerates to a maximum speed of 107 km/h before being forced to stop. The minimum distance to barrier was 0.12 m.
The free fall simulation was also successful: the drone accelerated to a top speed of 70 km/h downwards, before reaching a hover at a distance of 1.2 m above the barrier.

%% file: sections/5_results.tex
\subsection{Hardware setup}

Our quadrotor is built on a Chimera 7'' frame with four iFlight XING X2806.5 1300 KV brushless motors, a T-Motor F55A Pro II 4-in1 ESC, a MAMBA BASIC F722 Flight Controller (FC), a Teensy 4.1 microcontroller, a Vectornav VN-200 IMU+GPS, a FrSky R-XSR receiver, a DJI FPV air unit, and a Cadex FPV camera. We use a FrSky QX7 radio to send desired angular rates commands to Teensy microcontroller through the receiver. The VN-200 fuses GPS and IMU data with a built in extended Kalman filter. This data is sent to the Teensy as navigation data at 400 Hz. Using this data, the microcontroller then modifies these angular rates commands with the regulation function and then forwards them on to the FC. The FC runs betaflight, an open source software, to track the commanded angular rates. The PID loop runs at the gyro update rate at 8 kHz. The FC sends digital commands to the ESC using DSHOT600 at the same 8 kHz. FPV video is digitally streamed with a end to end latency of 25 ms from the DJI FPV air unit to DJI FPV goggles, which are worn by the operator.


\begin{figure}
    \centering
    \includegraphics[width=\columnwidth]{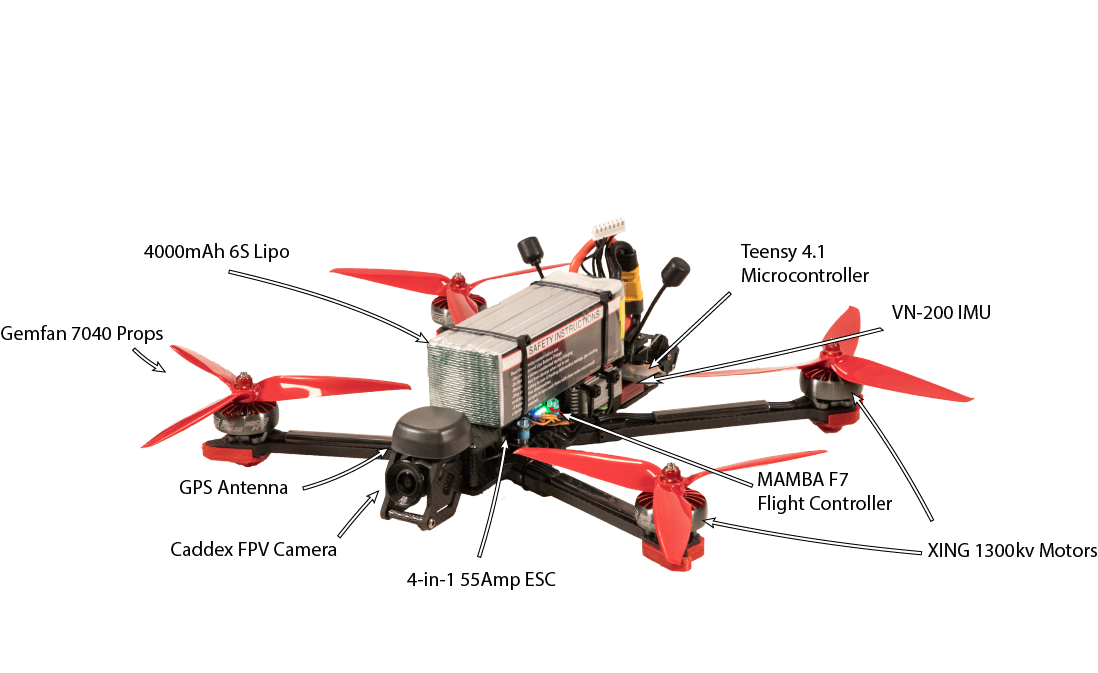}
    \caption{The 7" racing drone used for experiments.}
    \label{fig:hardware}
\end{figure}

\subsection{Filter software and execution}
To simplify the transition from \revf{simulation to hardware}, the barrier functions are first generated in MATLAB for simulation, and then codegen is used to create C++ code that will run on hardware.
All code used to generate and run the barrier functions on a Teensy 4.1 can be found at \url{https://github.com/DrewSingletary/racing_drone_geofencing}. This codebase also includes the interface for our specific receiver and flight controller, but this could be modified to fit other drone and radio configurations.

The execution time of the filter on the Teensy is approximately 350 $\mu$s. The algorithm runs at the update rate of the navigation data, which is 400 Hz. 


\subsection{Outdoor flight tests}

A large number of flight tests were done to verify the safety guarantees provided by our method. We emphasize two specific examples, but several more flight tests are displayed in Figure \ref{fig:barriers_hardware}.

\noindent \textbf{Test 1: Horizontal barrier at 104 km/h}. For this test, the pilot commanded the drone to head north at high speeds. The active component of the barrier was 40 m north of the initial position. The drone was able to reach a top speed of 104 km/h, before beginning to break at a distance of 15 m from the barrier. The pilot attempted to push the drone past the barrier, but was stopped at a minimum distance of 1.7 m from the barrier. Due to this, the pilot was then able to safely move away from the restricted airspace. 

Figure 7a showcases the results of the experiment. The results agree strongly with the simulation data in Fig. 5a, despite the human piloting commanding different desired inputs than the simulation. This is, in part, due to the very accurate angular rate tracking showcased at the bottom right of Fig. 7a. While there is a slight delay in this tracking, this is properly modeled in \eqref{eqn:omega_model}, thus it does not affect our ability to guarantee safety. 

Figure 7c shows the drone throughout this maneuver, highlighted in blue. Above this, the orientation of the drone at different snapshots are visualized. As shown, the drone reaches an angle of nearly 90$^{\circ}$ while breaking. 

\noindent \textbf{Test 2: Free fall from 70 m}. At the beginning of this test, the pilot was flying at an altitude of 70 m and then sends no commands, mimicking a loss of radio connection. The barrier was chosen to be a distance of 0.5 m from the ground. After a short free-fall, reaching a vertical velocity of -60 km/h, the safety filter stabilizes the drone before coming to a stop at a distance of 1.8 m above the ground. 

The data from this flight is visualized in Fig. 7b. Rather than plotting desired angular rates, we now showcase the desired throttle of the drone sent from the user compared to the throttle produced by the safety filter. Despite the pilot commanding no throttle for the entire duration of the descent, the drone is able to smoothly recover before crashing.

Again, when comparing this data to the simulation in Fig. 5b, notice the extremely similar results. In fact, the only major discrepancy, which is the fact that the simulation reached a speed of 10 km/h faster downwards than the drone, can be easily explained by a lack of drag in the simulation model. This did not occur during the horizontal tests, as the velocity controller is able to correct for this drag in flight.

\noindent \textbf{Testing for reliability and consistency.} Figure 8 highlights the reliability and consistency of this method in the application of geofencing. Four separate flights are plotted, two of which engage the horizontal barrier whereas two engage the vertical barrier. In each flight, the barrier is engaged two to four times, and every time, safety is maintained, and $\lambda$ never reaches zero, meaning the pilot never lost complete control of the drone for any period of time.

\begin{figure*}[t]
    \centering
    \begin{subfigure}{.49\textwidth}
    \includegraphics[trim={1cm 0 1.5cm 0},clip,width=\columnwidth]{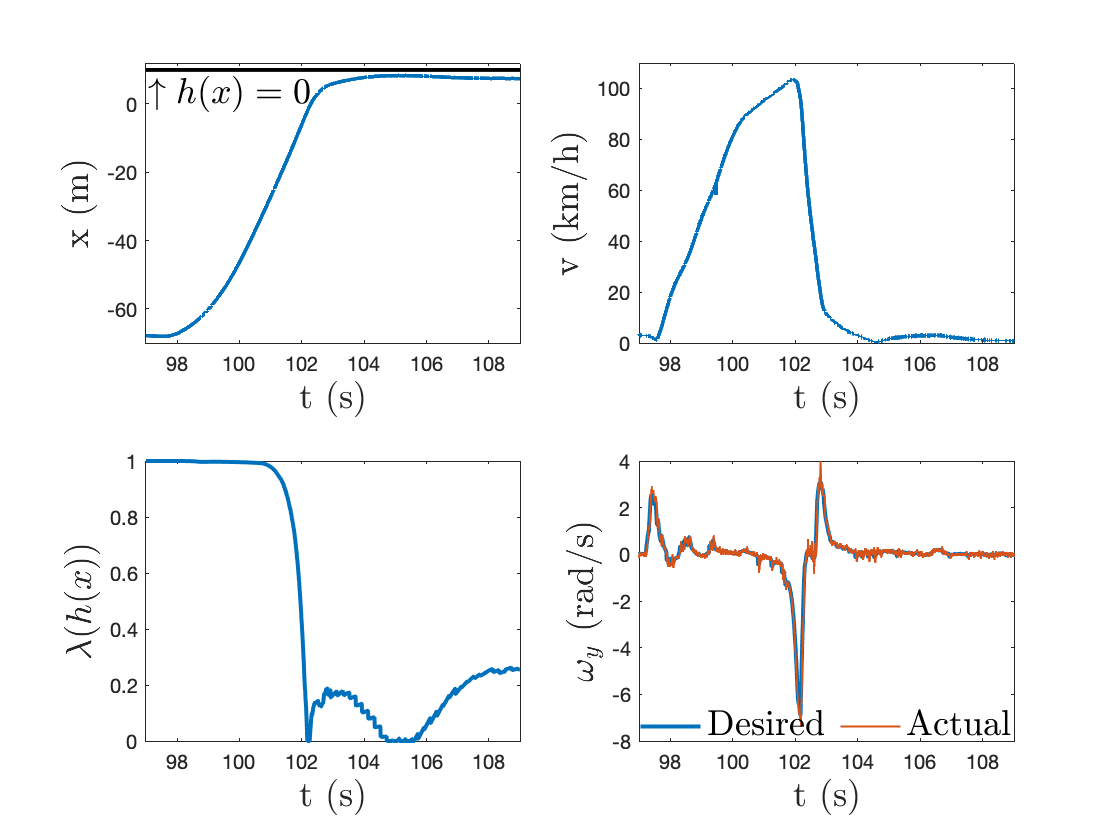}
    \caption{Horizontal barrier active at 105 km/h. }
    \end{subfigure}
    \rulesep
    \begin{subfigure}{.49\textwidth}
    \includegraphics[trim={1cm 0 1.5cm 0},clip,width=\columnwidth]{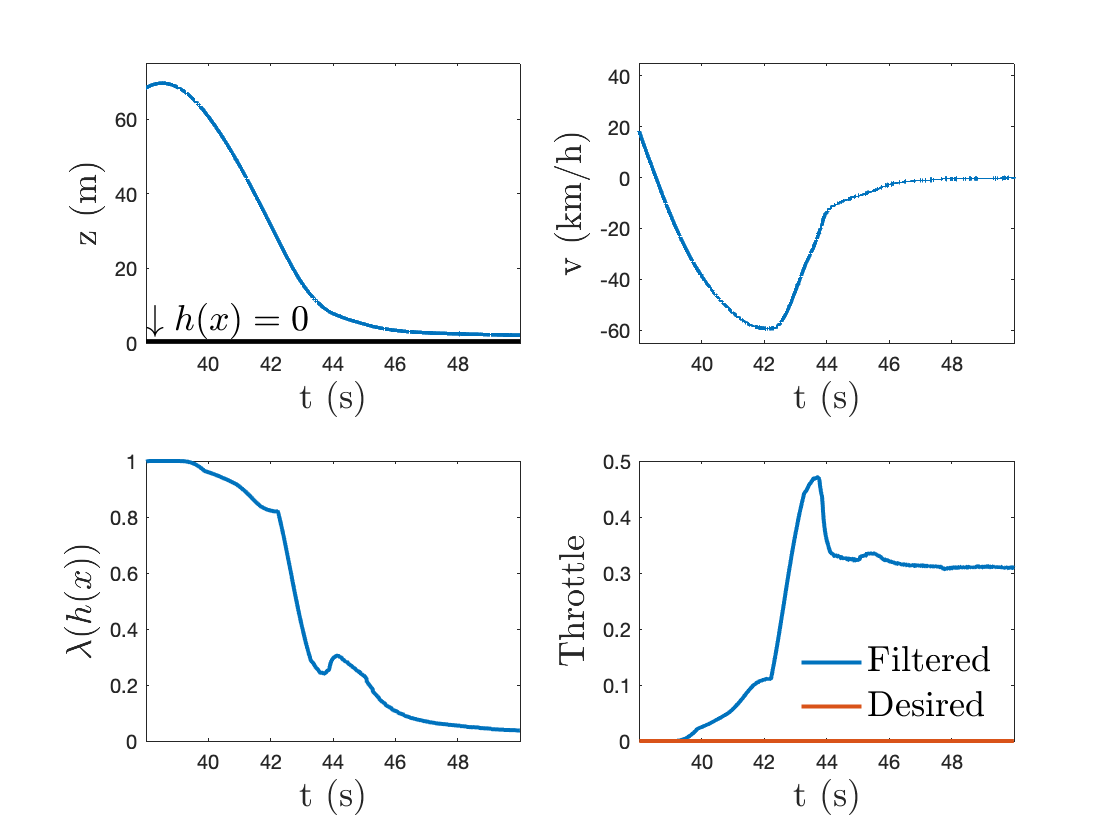}
    \caption{Vertical barrier active at -60 km/h. }
    \end{subfigure}
    \begin{subfigure}{1\textwidth}
    \includegraphics[width=0.586625\columnwidth]{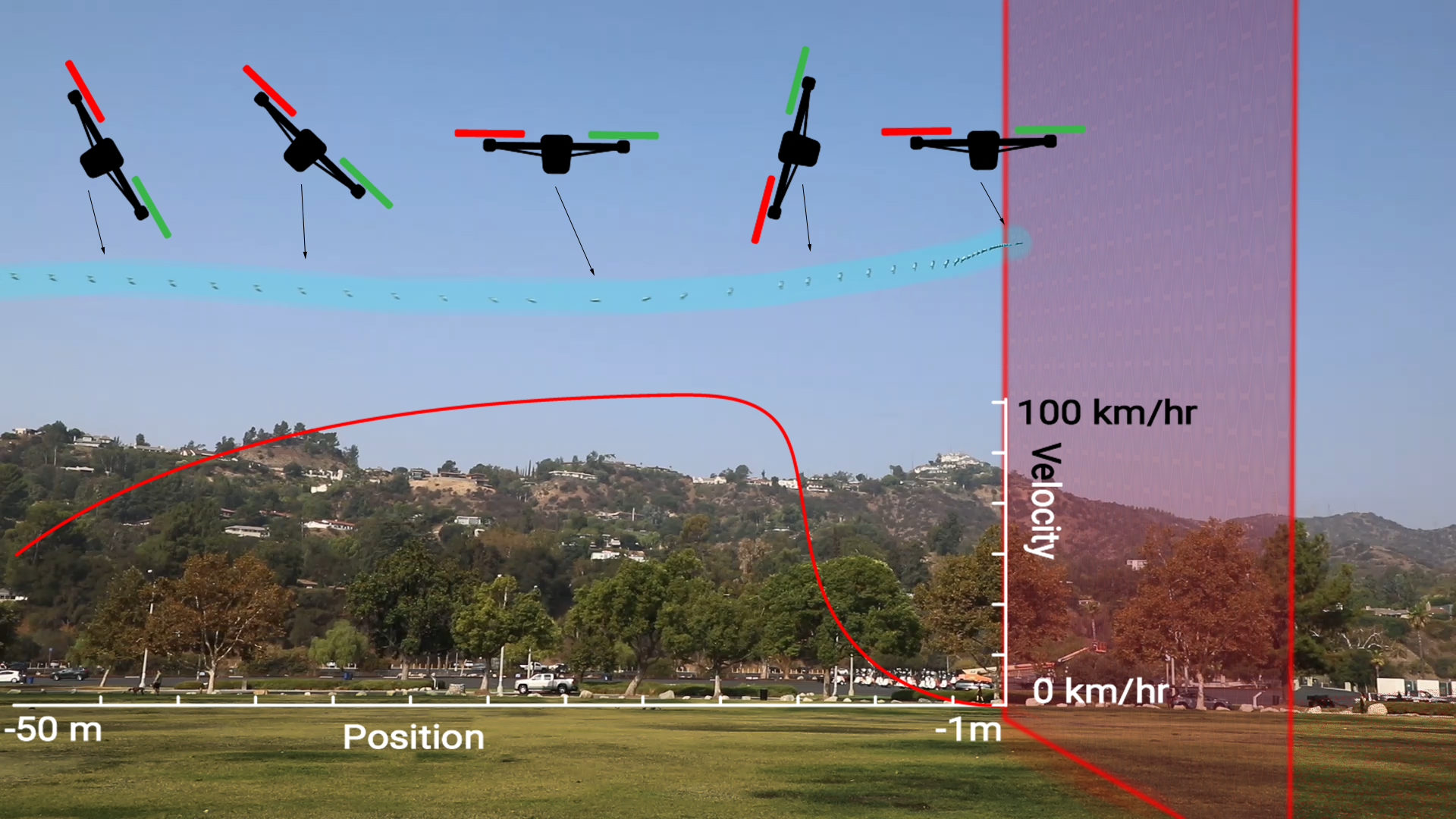}
    \rulesep
    \includegraphics[width=0.363375\columnwidth]{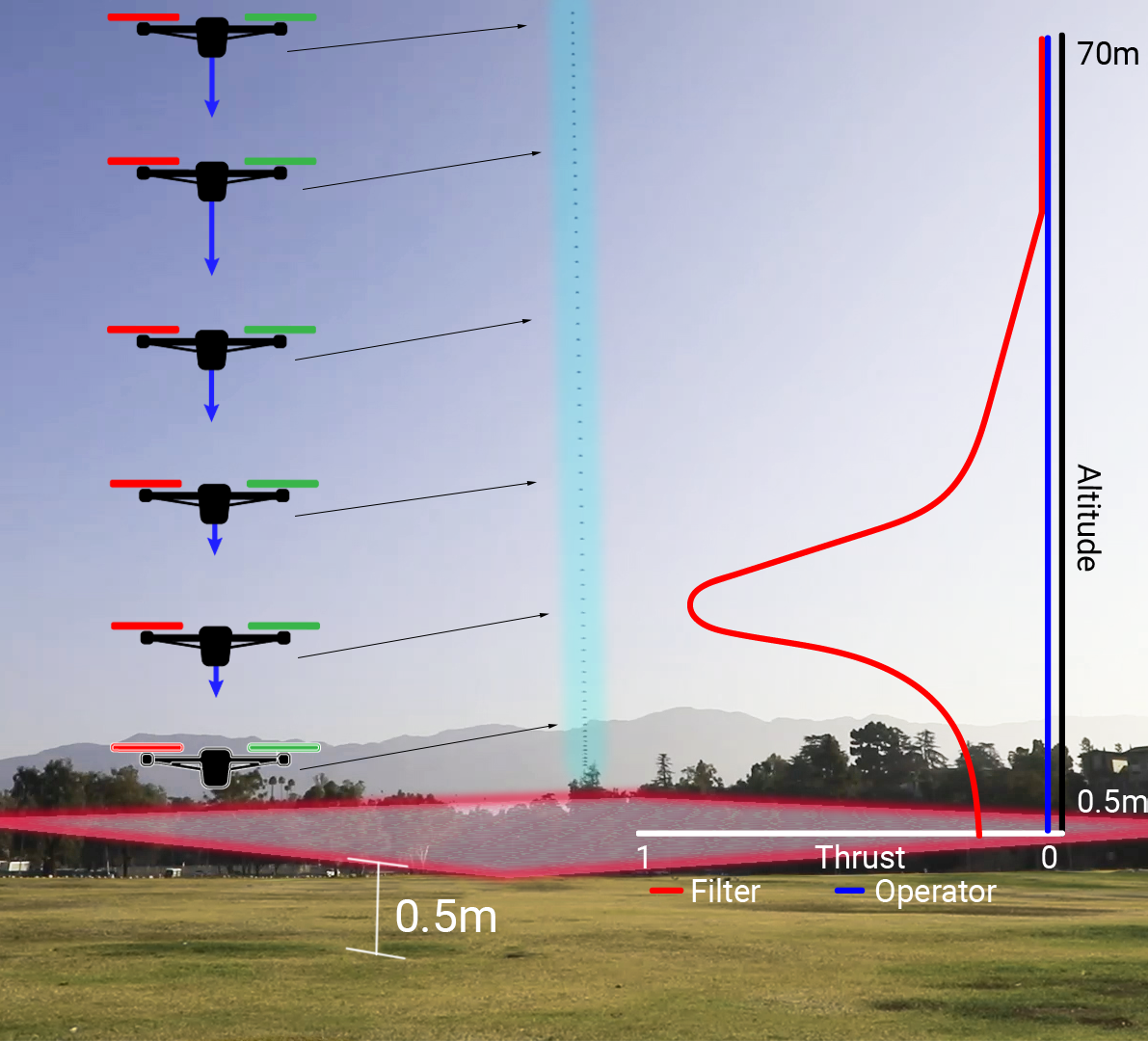}
    \caption{Actual drone flight during the two showcased example is highlighted in blue.}
    \end{subfigure}
    \caption{Two highlighted examples of geofencing with the high-speed racing drone. The video can be found at \url{https://youtu.be/_sCoAdBrgJw}}.
    \vspace{.5cm}
    \hrule
    \vspace{.5cm}
    \begin{subfigure}{.49\textwidth}
    \includegraphics[trim={2.9cm 0 1.3cm 0},clip,width=.99\columnwidth]{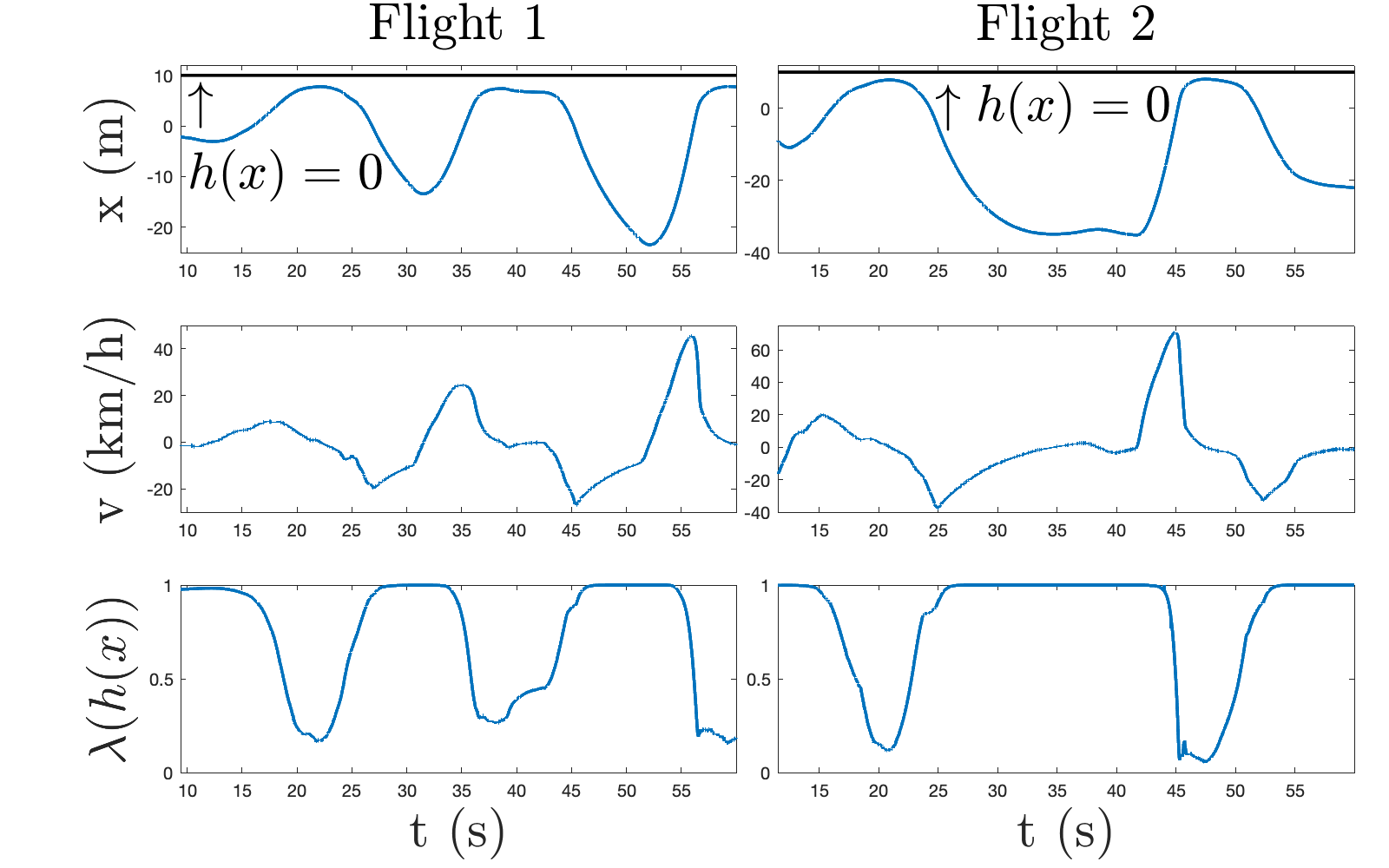}
    \caption{Horizontal barrier active}
    \end{subfigure}
        \rulesep
    \begin{subfigure}{.49\textwidth}
    \includegraphics[trim={2.9cm 0 1.3cm 0},clip,width=.99\columnwidth]{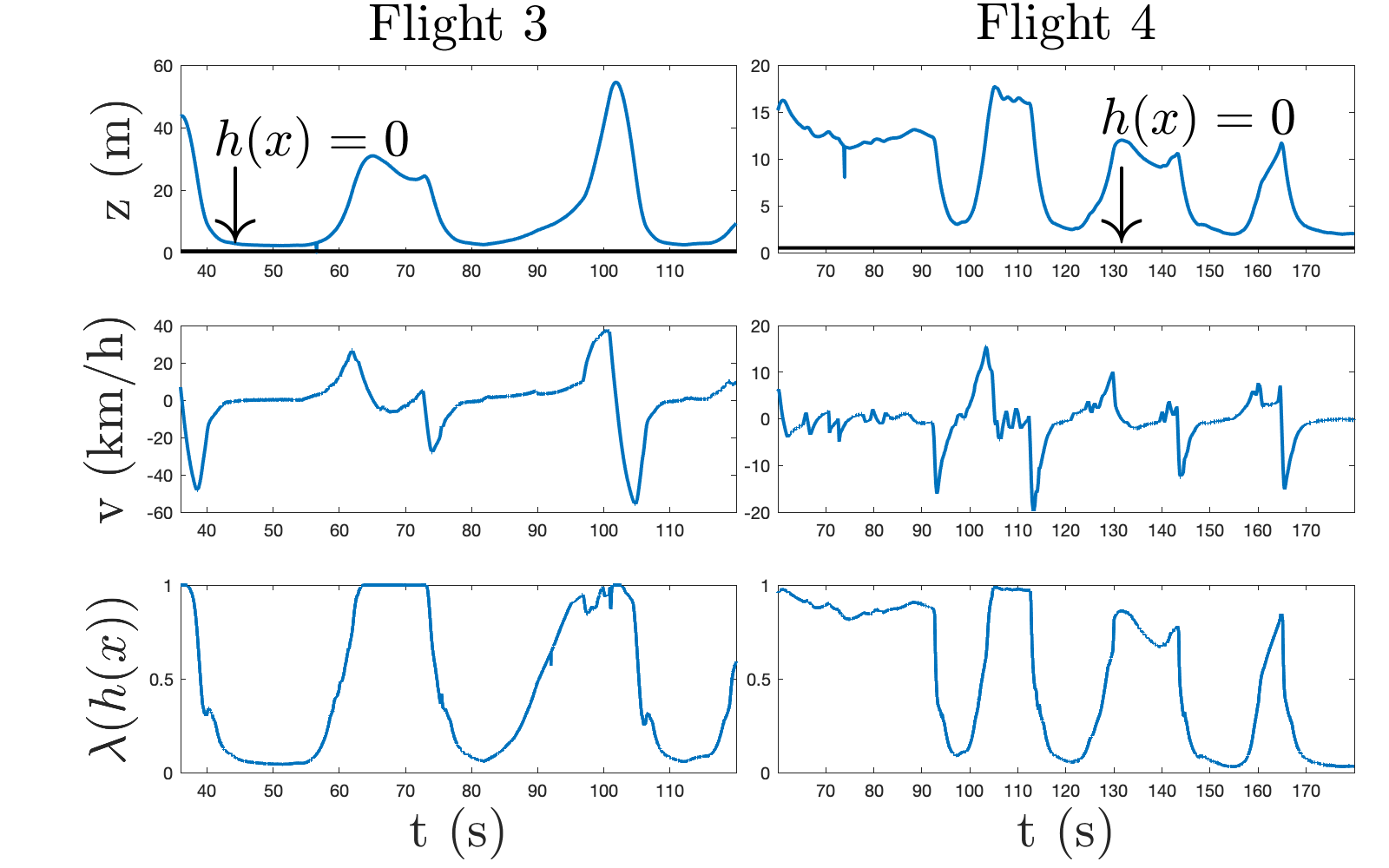}
    \caption{Vertical barrier active. }
    \end{subfigure}
    \caption{Four separate experimental runs where the drone is commanded to approach the barrier several times. $\lambda$ never reaches zero, meaning the pilot always has some amount of control, and the drone never leaves the defined safe set.}
    \label{fig:barriers_hardware}
\end{figure*}

%% file: sections/6_conclusion.tex
In this work, we showcased a novel safety filter intended to guarantee safe, high speed flight in the presence of a human operator. This method required no offline computations, and was implemented on a small microcontroller aboard a 7" racing drone. The filter was successful at keeping the drone inside of a desired region at speeds upwards of 100 km/h, and was easily able to recover from an 70 m free fall. Future work will consist of adding vision into the loop to dynamically construct the safe sets, as well as multi-robot collision avoidance at high-speeds for drone races.

%% file: main.bbl
\begin{thebibliography}{10}
\providecommand{\url}[1]{#1}
\csname url@samestyle\endcsname
\providecommand{\newblock}{\relax}
\providecommand{\bibinfo}[2]{#2}
\providecommand{\BIBentrySTDinterwordspacing}{\spaceskip=0pt\relax}
\providecommand{\BIBentryALTinterwordstretchfactor}{4}
\providecommand{\BIBentryALTinterwordspacing}{\spaceskip=\fontdimen2\font plus
\BIBentryALTinterwordstretchfactor\fontdimen3\font minus
  \fontdimen4\font\relax}
\providecommand{\BIBforeignlanguage}[2]{{%
\expandafter\ifx\csname l@#1\endcsname\relax
\typeout{** WARNING: IEEEtran.bst: No hyphenation pattern has been}%
\typeout{** loaded for the language `#1'. Using the pattern for}%
\typeout{** the default language instead.}%
\else
\language=\csname l@#1\endcsname
\fi
#2}}
\providecommand{\BIBdecl}{\relax}
\BIBdecl

\bibitem{foehn2020alphapilot}
P.~Foehn, D.~Brescianini, E.~Kaufmann, T.~Cieslewski, M.~Gehrig, M.~Muglikar,
  and D.~Scaramuzza, ``Alphapilot: Autonomous drone racing,'' \emph{arXiv
  preprint arXiv:2005.12813}, 2020.

\bibitem{jung2018perception}
S.~Jung, S.~Hwang, H.~Shin, and D.~H. Shim, ``Perception, guidance, and
  navigation for indoor autonomous drone racing using deep learning,''
  \emph{IEEE Robotics and Automation Letters}, vol.~3, no.~3, pp. 2539--2544,
  2018.

\bibitem{delmerico2019we}
J.~Delmerico, T.~Cieslewski, H.~Rebecq, M.~Faessler, and D.~Scaramuzza, ``Are
  we ready for autonomous drone racing? the {UZH-FPV} drone racing dataset,''
  in \emph{2019 International Conference on Robotics and Automation
  (ICRA)}.\hskip 1em plus 0.5em minus 0.4em\relax IEEE, 2019, pp. 6713--6719.

\bibitem{tordesillas2019real}
J.~Tordesillas, B.~T. Lopez, J.~Carter, J.~Ware, and J.~P. How, ``Real-time
  planning with multi-fidelity models for agile flights in unknown
  environments,'' in \emph{2019 International Conference on Robotics and
  Automation (ICRA)}.\hskip 1em plus 0.5em minus 0.4em\relax IEEE, 2019, pp.
  725--731.

\bibitem{tordesillas2019faster}
J.~Tordesillas, B.~T. Lopez, and J.~P. How, ``Faster: Fast and safe trajectory
  planner for flights in unknown environments,'' in \emph{2019 IEEE/RSJ
  International Conference on Intelligent Robots and Systems (IROS)}.\hskip 1em
  plus 0.5em minus 0.4em\relax IEEE, 2019, pp. 1934--1940.

\bibitem{kaufmann2018deep}
E.~Kaufmann, A.~Loquercio, R.~Ranftl, A.~Dosovitskiy, V.~Koltun, and
  D.~Scaramuzza, ``Deep drone racing: Learning agile flight in dynamic
  environments,'' in \emph{Conference on Robot Learning}.\hskip 1em plus 0.5em
  minus 0.4em\relax PMLR, 2018, pp. 133--145.

\bibitem{kaufmann2019beauty}
E.~Kaufmann, M.~Gehrig, P.~Foehn, R.~Ranftl, A.~Dosovitskiy, V.~Koltun, and
  D.~Scaramuzza, ``Beauty and the beast: Optimal methods meet learning for
  drone racing,'' in \emph{2019 International Conference on Robotics and
  Automation (ICRA)}.\hskip 1em plus 0.5em minus 0.4em\relax IEEE, 2019, pp.
  690--696.

\bibitem{song2021autonomous}
Y.~Song, M.~Steinweg, E.~Kaufmann, and D.~Scaramuzza, ``Autonomous drone racing
  with deep reinforcement learning,'' \emph{arXiv preprint arXiv:2103.08624},
  2021.

\bibitem{moon2019challenges}
H.~Moon, J.~Martinez-Carranza, T.~Cieslewski, M.~Faessler, D.~Falanga,
  A.~Simovic, D.~Scaramuzza, S.~Li, M.~Ozo, C.~De~Wagter \emph{et~al.},
  ``Challenges and implemented technologies used in autonomous drone racing,''
  \emph{Intelligent Service Robotics}, vol.~12, no.~2, pp. 137--148, 2019.

\bibitem{broad2019highly}
A.~Broad, T.~Murphey, and B.~Argall, ``Highly parallelized data-driven mpc for
  minimal intervention shared control,'' \emph{arXiv preprint
  arXiv:1906.02318}, 2019.

\bibitem{tearle2021predictive}
B.~Tearle, K.~P. Wabersich, A.~Carron, and M.~N. Zeilinger, ``A predictive
  safety filter for learning-based racing control,'' \emph{arXiv preprint
  arXiv:2102.11907}, 2021.

\bibitem{zhang2017model}
S.~Zhang, D.~Wei, M.~Q. Huynh, J.~X. Quek, X.~Ma, and L.~Xie, ``Model
  predictive control based dynamic geofence system for unmanned aerial
  vehicles,'' in \emph{AIAA Information Systems-AIAA Infotech@ Aerospace},
  2017, p. 0675.

\bibitem{hermand2018constrained}
E.~Hermand, T.~W. Nguyen, M.~Hosseinzadeh, and E.~Garone, ``Constrained control
  of uavs in geofencing applications,'' in \emph{2018 26th Mediterranean
  Conference on Control and Automation (MED)}.\hskip 1em plus 0.5em minus
  0.4em\relax IEEE, 2018, pp. 217--222.

\bibitem{ames2017cbf}
A.~D. Ames, X.~Xu, J.~W. Grizzle, and P.~Tabuada, ``Control barrier function
  based quadratic programs for safety critical systems,'' \emph{IEEE
  Transactions on Automatic Control}, vol.~62, no.~8, pp. 3861--3876, 2017.

\bibitem{ames2019control}
A.~D. Ames, S.~Coogan, M.~Egerstedt, G.~Notomista, K.~Sreenath, and P.~Tabuada,
  ``Control barrier functions: Theory and applications,'' in \emph{2019 18th
  European Control Conference (ECC)}.\hskip 1em plus 0.5em minus 0.4em\relax
  IEEE, 2019, pp. 3420--3431.

\bibitem{wang2018safe}
L.~Wang, E.~A. Theodorou, and M.~Egerstedt, ``Safe learning of quadrotor
  dynamics using barrier certificates,'' in \emph{2018 IEEE International
  Conference on Robotics and Automation (ICRA)}.\hskip 1em plus 0.5em minus
  0.4em\relax IEEE, 2018, pp. 2460--2465.

\bibitem{chen2020guaranteed}
Y.~Chen, A.~Singletary, and A.~D. Ames, ``Guaranteed obstacle avoidance for
  multi-robot operations with limited actuation: A control barrier function
  approach,'' \emph{IEEE Control Systems Letters}, vol.~5, no.~1, pp. 127--132,
  2020.

\bibitem{singletary2020safety}
A.~Singletary, T.~Gurriet, P.~Nilsson, and A.~D. Ames, ``Safety-critical rapid
  aerial exploration of unknown environments,'' in \emph{2020 IEEE
  International Conference on Robotics and Automation (ICRA)}.\hskip 1em plus
  0.5em minus 0.4em\relax IEEE, 2020, pp. 10\,270--10\,276.

\bibitem{aubin2009viability}
J.-P. Aubin, \emph{Viability theory}.\hskip 1em plus 0.5em minus 0.4em\relax
  Springer Science, 2009.

\bibitem{mitchell2005time}
I.~M. Mitchell, A.~M. Bayen, and C.~J. Tomlin, ``A time-dependent
  {Hamilton}-{Jacobi} formulation of reachable sets for continuous dynamic
  games,'' \emph{IEEE Transactions on Automatic Control}, vol.~50, no.~7, pp.
  947--957, 2005.

\bibitem{gurriet2018online}
T.~Gurriet, M.~Mote, A.~D. Ames, and E.~Feron, ``An online approach to active
  set invariance,'' in \emph{2018 IEEE Conference on Decision and Control
  (CDC)}.\hskip 1em plus 0.5em minus 0.4em\relax IEEE, 2018, pp. 3592--3599.

\bibitem{chen2021backup}
Y.~Chen, M.~Jankovic, M.~Santillo, and A.~D. Ames, ``Backup control barrier
  functions: Formulation and comparative study,'' \emph{arXiv preprint
  arXiv:2104.11332}, 2021.

\bibitem{5717652}
T.~Lee, M.~Leok, and N.~H. McClamroch, ``Geometric tracking control of a
  quadrotor {UAV} on {SE(3)},'' in \emph{49th IEEE Conference on Decision and
  Control (CDC)}, 2010, pp. 5420--5425.

\end{thebibliography}
